%% file: main.tex
\documentclass[10pt]{article}
\pdfoutput=1
\usepackage[margin=1in]{geometry}
\usepackage{amsmath,array,bm}
\usepackage{amsmath,amsfonts}
\usepackage{amsthm}
\usepackage{amssymb,latexsym}
\usepackage{algorithm}
\usepackage{subcaption}
\usepackage{algorithmicx}
\usepackage{marginnote}
\usepackage[]{algpseudocode}
\usepackage[dvipsnames]{xcolor}
\usepackage{graphicx}
\usepackage{wrapfig, framed, caption}
\usepackage{caption}
\usepackage{float}
\usepackage{thmtools}
\usepackage{hyperref}
\hypersetup{pdftex, plainpages = false, pdfpagelabels, 
	bookmarks=true,
	bookmarksopen = true,
	bookmarksnumbered = true,
	breaklinks = true,
	linktocpage,
	pagebackref,
	colorlinks = true,  
	linkcolor = blue,
	urlcolor  = blue,
	citecolor = red,
	anchorcolor = green,
	hyperindex = true,
	hyperfigures
}
\usepackage{mathtools}
\usepackage{thm-restate}
\usepackage{cleveref}
\usepackage{bussproofs}
\usepackage[square,sort,comma,numbers]{natbib}
\usepackage{tikz}
\usepackage{xspace}
\usepackage{xcolor}
\usepackage{tipa}
\usepackage{mdframed}
\usepackage[export]{adjustbox}[2011/08/13]
\usepackage{multirow}
\usepackage{hhline}
\usepackage{todonotes}
\usepackage{bbm}
\usepackage{mathrsfs}
\usepackage{tcolorbox} 

\newtheorem{lemma}{Lemma}
\newtheorem{claim}{Claim}

\newtheorem{theorem}{Theorem}

\newtheorem{definition}{Definition}
\newtheorem{proposition}{Proposition}

\input{macro}

\title{FPT Approximations for Capacitated/Fair Clustering with Outliers} 
\author{Rajni Dabas \thanks{Department of Computer Science, University of Delhi, New Delhi, India. Email: \texttt{rajni@cs.du.ac.in}.} \and Neelima Gupta \thanks{Department of Computer Science, University of Delhi, New Delhi, India. Email: \texttt{ngupta@cs.du.ac.in}.} \and Tanmay Inamdar \thanks{Department of Informatics, University of Bergen, Bergen, Norway. Email: \texttt{Tanmay.Inamdar@uib.no.}\ \ The research leading to these results has received funding the European Research Council (ERC) via grant LOPPRE, reference 819416.}}
\date{}

\begin{document}
\maketitle

\begin{abstract}
Clustering problems such as $k$-\textsc{Median}, and $k$-\textsc{Means}, are motivated from applications such as location planning, unsupervised learning among others. In many such applications, it is important to find the clustering of points that is not ``skewed'' in terms of the number of points, i.e., no cluster should contain \emph{too many} points. This is often modeled by introducing \emph{capacity constraints} on the sizes of clusters. In an orthogonal direction, another important consideration in the domain of clustering is how to handle the presence of \emph{outliers} in the data. Indeed, the aforementioned clustering problems have been generalized in the literature to separately handle capacity constraints and outliers. However, to the best of our knowledge, there has been very little work on studying the approximability of clustering problems that can simultaneously handle capacity constraints as well as outliers.

We bridge this gap and initiate the study of the \textsc{Capacitated $k$-Median with Outliers} (\ckmo) problem. In this problem, we want to cluster all except $m$ \emph{outlier points} into at most $k$ clusters, such that (i) the clusters respect the capacity constraints, and (ii) the cost of clustering, defined as the sum of distances of each \emph{non-outlier} point to its assigned cluster-center, is minimized.

We design the first constant-factor approximation algorithms for \ckmo. In particular, our algorithm returns a $(3+\epsilon)$-approximation for \ckmo in general metric spaces, and a $(1+\epsilon)$-approximation in Euclidean spaces of constant dimension, that runs in time in time $f(k, m, \epsilon) \cdot |\cI|^{O(1)}$, where $|\cI|$ denotes the input size. We are also able to extend these results to obtain similar FPT approximations for a broader class of problems, including Capacitated $k$-Means/$k$-Facility Location with Outliers, and Size-Balanced Fair Clustering problems with Outliers. For each of these problems, we obtain an approximation ratio that matches the best known guarantee of the corresponding outlier-free problem.
\end{abstract}

\input{intro}

\input{prelims}

\input{ckmo}
\input{conclusion}
\bibliography{ref}
\end{document}

%% file: macro.tex
\newcommand{\points}{P}
\newcommand{\dist}{d}
\newcommand{\clients}{C}
\newcommand{\facilities}{\mathcal{F}}

\newcommand{\solution}{F}
\newcommand{\assign}{\sigma}
\newcommand{\capacity}[1]{u_{#1}}
\newcommand{\outliers}{m}
\newcommand{\cost}{\mathsf{cost}}

\newcommand{\wt}{w}
\newcommand{\wtset}{W}
\newcommand{\wcost}{\mathsf{wcost}}
\newcommand{\dummy}{d_f}
\newcommand{\opt}{\mathsf{OPT}}
\newcommand{\rad}{R}
\newcommand{\ceil}[1]{\lceil #1\rceil}

\newcommand{\ball}{\textsf{Ball}}
\newcommand{\ringf}[1]{\clients_{f,#1}}
\newcommand{\rings}[2]{\clients_{#1,#2}}
\newcommand{\clusterf}{\clients_f}
\newcommand{\ring}{\clients_j}

\newcommand{\radr}{\rad_j}
\newcommand{\vectorx}{\mathbb{X}}
\newcommand{\vectory}{\mathbb{Y}}

\newcommand{\samples}[2]{S_{#1,#2}}
\newcommand{\func}{g}
\newcommand{\ex}{\mathbb{E}}
\newcommand{\flow}{\mathsf{FI}}
\newcommand{\out}{\clients^o}

\newcommand{\samplef}{S^f}
\newcommand{\flowS}{\phi'}
\newcommand{\lr}[1]{\left(#1\right)}
\newcommand{\LR}[1]{\left\{#1\right\}}
\newcommand{\red}[1]{{\color{red}#1}}

\newcommand{\poly}{\mathsf{poly}}
\newcommand{\nat}{\mathbb{N}}

\newcommand{\vecz}{\mathbbm{z}}
\newcommand{\vecv}{\mathbbm{v}}
\newcommand{\vecone}{\mathbbm{1}}
\newcommand{\I}{I}
\newcommand{\cI}{I_\outliers}
\newcommand{\IwO}{I_{0}}
\newcommand{\wtpoints}{S}
\newcommand{\wtout}{T}

\newcommand{\largedist}{D}
\newcommand{\order}{\sigma}
\newcommand{\fl}{\phi}
\newcommand{\groups}{\ell}
\newcommand{\fairone}{\alpha}
\newcommand{\fairtwo}{\beta}
\newcommand{\wckm}{\textsf{WC$k$M}\xspace}
\newcommand{\wckmo}{\textsf{WC$k$MO}\xspace}
\newcommand{\mcf}{\textsf{MCF}\xspace}
\newcommand{\mcfo}{\textsf{MCFO}\xspace}
\newcommand{\kmed}{$k$-\textsc{Median}\xspace}
\newcommand{\kmeans}{$k$-\textsc{Means}\xspace}
\newcommand{\kcenter}{$k$-\textsc{Center}\xspace}
\newcommand{\facloc}{\textsc{Facility Location}\xspace}
\newcommand{\NP}{\textsf{NP}\xspace}
\newcommand{\ckm}{\textsf{C$k$M}\xspace}
\newcommand{\ckmo}{\textsf{C$k$MO}\xspace}
\newcommand{\kmo}{\textsf{$k$MO}\xspace}
\newcommand{\factor}{\gamma}
\newcommand{\consta}{\lambda_1}
\newcommand{\constb}{\lambda_2}
\newcommand{\collection}{\mathcal{M}}
\newcommand{\wfao}{WFAO\xspace}
\newcommand{\constr}{$\mathscr{C}$\xspace}
\newcommand{\factortwo}{\zeta}

%% file: intro.tex

\section{Introduction} \label{sec:intro}

Clustering problems such as \kmed, \kmeans, \kcenter, and \facloc are among the most well-studied problems in the approximation algorithms literature. Although in general, these problems are \NP-hard, constant factor approximations are known for each of them in polynomial time \cite{GowdaPST23,AhmadianNSW20,hochbaum1982approximation,Li2011FL}. Furthermore, in several special settings, e.g., when the points belong to low-dimensional euclidean spaces, it is possible to obtain near-optimal approximation ratios~\cite{FriggstadRS19,DBLP:journals/siamcomp/Cohen-AddadKM19,Arora,AhmadianNSW20,Cohen-Addad20Doubling,Cohen-AddadSS21}.

In many applications, it is desirable to find a clustering of the given point set with certain additional properties. Consequently, there has been a considerable effort in studying different \emph{constrained} versions of aforementioned clustering problems. Notable examples include clustering to handle capacities, outliers, fault tolerance, fairness, knapsack and matroid constraints on centers, among others. Typically, handling each such constraint requires a different set of techniques, and there is a varying degree of success in obtaining good approximation guarantees for each of these problems~\cite{AnMP2015,Bansal,Anfocs2014,CharikarKMN01,KrishnaswamyLS18,Chierichetti0LV17,Chen2016Matroidknapsack}. However, it is quite natural to seek a solution that satisfies multiple constraints at the same time, which motivates the following question.

\begin{tcolorbox}
	Can we design approximation algorithms to handle multiple, orthogonal constraints in clustering problems such as \kmed? 
\end{tcolorbox}
In this paper, we focus on handling capacity constraints (resp.\ fairness) along with outliers in the context of \kmed, \kmeans and related objectives. Before discussing the problem studied in this paper, we build towards its definition by discussing several special cases thereof, that have been extensively studied in the literature. 

The general setup of \kmed, and its generalizations is as follows. We are given an instance $\I = ((\points, \dist), \clients, \facilities, k)$, where $(\points, \dist)$ is a metric space, $\clients \subseteq \points$ is a set of $n$ \emph{clients}, $\facilities \subseteq \points$ is a set of (candidate) facilities, and $k$ is a positive integer. In \kmed, we want to find a set of facilities $\solution \subseteq \facilities$ of size $k$, minimizing the objective $\sum_{c \in \clients} \min_{f \in \solution} \dist(c, f)$. There is a long series of approximation algorithms for \kmed, culminating in $2.613$-approximation in general metrics in polynomial time \cite{GowdaPST23}; whereas approximation schemes are known in low-dimensional euclidean and doubling spaces, and metrics induced by minor-free graph classes \cite{DBLP:journals/siamcomp/Cohen-AddadKM19,FriggstadRS19,Cohen-AddadSS21}. 
\textsc{$k$-Median with Outliers} (\kmo) is a generalization of \kmed, where we are also given a non-negative integer $m$, denoting the number of outliers, whose contribution of distances can be excluded from the objective. In this case, the minimization is over all possible choices of $m$-subsets for outliers, and $k$-subsets for the facilities. Constant approximations are known for \kmo in general metrics via iterative rounding \cite{KrishnaswamyLS18,GuptaMZ21}. 

\textsc{Capacitated $k$-Median} (\ckm) introduces capacity constraints into the vanilla \textsc{$k$-Median} problem. Here, for each candidate facility $f \in \facilities$, we are given a non-negative integer $u_f$ that denotes its capacity. If we decide to include $f$ in our solution of size $k$, we can assign at most $u_f$ clients to $f$. The goal is to find a solution $\solution \subseteq \facilities$ of size $k$, and a capacity-respecting assignment, such that the sum of distances of each client to its assigned facility, is minimized. The presence of capacity constraints makes the problem significantly harder -- indeed obtaining a constant-factor approximation for \ckm in general metric spaces, has been a long-standing open question in the area of approximation algorithms. There is a partial progress toward this goal by designing so-called \emph{bi-criteria} approximations that approximate the cost up to a constant factor, while also violating either the ``$k$-constraint'' or the capacity constraints by a small factor~\cite{Charikar:1999,ByrkaRybicki2015,ChuzhoyR05,Demirci2016,capkmshanfeili2014,capkmByrkaFRS2013,GroverGKP18}. 

In this paper, we are interested in studying clustering problems that simultaneously handles both the capacity constraints and outliers. To the best of our knowledge, the only works that simultaneously handle both of these constraints is \kcenter and \textsc{Facility Location} objectives, see e.g., Cygan et al.~\cite{Cygan-kCO} and Dabas and Gupta~\cite{DabasG22}. In this paper, we study a common generalization of \kmo and \ckm. We call this problem \textsc{Capacitated $k$-Median with Outliers} (\ckmo), defined as follows. 

\begin{tcolorbox}[colback=white!5!white,colframe=gray!75!black]
	\ckmo
	\\\textbf{Input:} Instance $\cI= ((\points, \dist), \clients, \facilities, k, \capacity{}, \outliers)$, where $(\points, \dist), \clients, \facilities$ have meaning as defined above. Additionally, each facility $f \in \facilities$ has a capacity $\capacity{f} \in \mathbb{N}$, and $0 \le \outliers \le n$ denotes the number of clients that can be left unassigned. 
	\\\textbf{Task:} To find
	\begin{itemize}
		\item a subset $\solution \subseteq \facilities$ of size at most $k$,
		\item a subset $\clients' \subseteq \clients$ of size at most $\outliers$ (called the set of \emph{outliers}) and
		\item an assignment $\assign : (\clients \setminus \clients') \rightarrow \solution$ respecting the capacities, i.e., for each facility $ f \in \facilities$, it holds that $|\sigma^{-1}(f)| \leq \capacity{f}$
	\end{itemize}
	such that the cost $\displaystyle \sum_{c \in \clients \setminus \clients'} \dist(c,\sigma(c))$ is minimized.
\end{tcolorbox}
Note that \ckmo generalizes \ckm, and the polynomial-time approximability of the latter itself remains open. A recent direction in the area is to look beyond polynomial-time approximations, in order to break the barrier of constant-factor approximability, or to obtain tight approximation factors.

\medskip\noindent\textbf{\textsf{FPT Approximations}.} By relaxing the strict requirement of polynomial running time, it is possible to obtain improved approximations for \kmed and its variants. Cohen-Addad et al. \cite{DBLP:conf/icalp/Cohen-AddadG0LL19} designed an $(1+2/e+\epsilon)$-approximation for \kmed that runs in time FPT in $k$ and $\epsilon$ \footnote{An algorithm with the running time $f(p) \cdot |I|^{O(1)}$ is said to be \emph{fixed-parameter tractable} (FPT) in the parameter $p$, where $|I|$ denotes the input size, and $f$ is an arbitrary function.}, matching a complexity-theoretic lower bound. Very recently, Agrawal et al.~\cite{agrawal2023kMO} designed an approximation-preserving reduction from \kmo to \kmed that is FPT in $k, m$ and $\epsilon$ (recall that $m$ denotes the number of outliers). As a corollary, they obtained a $(1+2/e+\epsilon)$-approximation for $k$MO that is FPT in $k, m$ and $\epsilon$. Finally, in the context of \ckm, Adamczyk et al.~\cite{AdamczykBMM019} broke the barrier of constant-factor approximability in general metrics, by designing first constant-factor approximation that runs in time FPT in $k, \epsilon$. This was later improved by Cohen-Addad and Li \cite{Cohen-AddadL19Capacitated} to a $(3+\epsilon)$-approximation for \ckm in general metrics, and a $(1+\epsilon)$-approximation in low-dimensional Euclidean metrics, in time in FPT in $k$ and $\epsilon$. There has also been some recent work on designing FPT approximations for capacitated versions of other related clustering objectives, such as sum-of-radii \cite{IV20,BLS23}. For the comparably easier objective of $k$-\textsc{Center}, Goyal and Jaiswal \cite{GoyalJ23} designed FPT approximations parameterized by $k$ and $m$, that handle different constraints including capacities and outliers. Inspired from this success of FPT approximations, we continue this line of work in the context of \ckmo and related problems.

\paragraph{Recent Independent Work.} In a recent and independent work, Jaiswal and Kumar \cite{jaiswal2023clustering} design similar \emph{almost approximation-preserving} reductions from constrained \kmed/\kmeans problems with outliers, to their outlier-free counterparts. As a result, they obtain similar approximations for various constrained versions of $k$-\textsc{Median/Means with Outliers} for various constraints, including capacities, $(\alpha, \beta)$-fairness, among others. However, a cursory examination suggests that their techniques are quite different from ours.

\subsection{Our Results and Techniques} 
One of the main results of this paper is the following theorem.
\begin{theorem}[Informal]
\label{thm:kmed-approx-preserving} 
There exists an approximation-preserving reduction from \ckmo to \ckm that runs in time FPT in $k, m$, and $\epsilon$, where the underlying metric space remains unchanged. 
\end{theorem}

We can also extend this result for related clustering objectives, such as \kmeans and \textsc{$k$-Facility Location}. Recall that \kmeans objective function involves sum of squares of distances, whereas $k$-\textsc{Facility Location} is a generalization of \kmed, where the objective involves facility opening costs. For both of these objectives---and indeed for $(k, z)$-clustering that involves $z$-th power of distances---we can reduce the capacitated clustering with outliers problem to the corresponding capacitated clustering variant (without outliers) in time FPT in $k, m$ and $\epsilon$. Then, by plugging in the best known approximations for the capacitated clustering problems (without outliers), we obtain matching approximation ratios. These corollaries are stated in \Cref{fig:corollaries}.

\begin{figure}
	\centering
	\begin{tabular}{|c|c|c|c|c|}
		\hline
		Problem & Metric & Approximation  \\
		\hline
		{\footnotesize \textsc{\constr $k$-Median with Outliers}} & Arbitrary & $3+\epsilon$ 
		\\ \hline
		{\footnotesize \textsc{\constr $k$-Median with Outliers}} & Euclidean & $1+\epsilon$  \\\hline
		{\footnotesize \textsc{\constr $k$-Means with Outliers}} & Arbitrary & $9+\epsilon$  \\\hline
		{\footnotesize \textsc{\constr $k$-Facility Location with Outliers}} & Arbitrary & $3+\epsilon$  \\\hline
	\end{tabular}
	\caption{Selected corollaries of \Cref{thm:kmed-approx-preserving}, \Cref{thm:fairkmed-approx-informal} and their extensions. In the first column, the constraint \constr can be replaced with either \textsc{Capacitated} or with \textsc{($\fairone,\fairtwo$)-Fair}. In each row, when \constr = \textsc{Capacitated}, we plug in the relevant FPT approximation from Cohen-Addad and Li \cite{Cohen-AddadL19Capacitated}, whereas when \constr = \textsc{($\fairone,\fairtwo$)-Fair}, we use the relevant FPT approximation from Bandyapadhyay et al.~\cite{BandyapadhyayFS21}; and obtain the approximation ratios as stated in the third column. 
		\\Finally, by plugging in the QPTAS from \cite{Cohen-AddadL19Capacitated} for \ckm in doubling spaces, we obtain a $(1+\epsilon)$-approximation for \ckmo that runs in time $f(k, m, \epsilon) \cdot 2^{\poly(\log(|\cI|, 1/\epsilon))}$.} \label{fig:corollaries}
\end{figure}

\medskip\noindent\textbf{\textsf{Ring sampling approach and comparison to related work}.}
A starting point of our algorithm---like many other FPT approximations for \kmed and variants (\cite{DBLP:conf/icalp/Cohen-AddadG0LL19,Cohen-AddadL19Capacitated,agrawal2023kMO})---is the elegant ``ring sampling'' approach of Chen \cite{chen2009coresets}. We start with a brief overview of this construction, and subsequently discuss how this is used in our reduction framework to handle capacities and outliers.

In this construction, the idea is to do random sampling for ``compressing'' the given set of $n$ clients into a ``small'' number of \emph{weighted clients} that is a good approximation of the original set. \footnote{This notion can be formalized to define \emph{coresets}. We omit a formal definition here.} Let us consider the simplest setting of \kmed without capacities or outliers as \cite{chen2009coresets}. We start from a crude (constant factor) approximation $\solution' \subseteq \facilities$ for \kmed, and use it to partition the clients into \emph{concentric rings} around the centers in $\solution'$. Here, a ring is a subset of clients that are within a distance $r$ and $2r$ from some $f \in \solution'$, where we consider geometrically increasing radii $r$. 
Next, we take a uniform sample of clients from each ring of large enough size, and assign appropriate weight to each sampled client. The resulting weighted sample has size $O(k \log n/\epsilon)^2)$. It can be shown that, the \emph{weighted clients} in the sample approximately preserve the cost of clustering w.r.t.~any solution with high probability. This is shown using a Chernoff-Hoeffding type concentration bound, which uses the following observation. Even though the distances of each weighted client in the sample is a random variable, all random variables lie in a bounded interval due to triangle inequality. Precisely, for any solution $\solution \subseteq \facilities$, and for any ring $R \subseteq \clients$ with radius $r$, the distances $\LR{d(c, \solution)}_{c \in R}$ lie within an interval of length $2r$. 

Since the size of the sample is relatively small, it is used in \cite{DBLP:conf/icalp/Cohen-AddadG0LL19} for partial enumeration of a subset of approximate solutions for \kmed, which leads to the FPT running time. Agrawal et al.~\cite{agrawal2023kMO} argue that in the context of \kmo, by taking an even larger sample size (that now also depends on $\outliers$), the concentration bound is robust enough to handle outliers. In particular, they show that for each ring, the concentration bound continues to hold \emph{even after} ignoring the contribution of at most $\outliers$ points that are farthest from the solution. 

Cohen-Addad and Li \cite{Cohen-AddadL19Capacitated} extend the ring sampling approach to \ckm, but the analysis is more intricate due to the following reason. Consider a particular ring $R$ with center $f$, and a hypothetical solution $\solution \subseteq \facilities$ of size $k$. For each $f_i \in \solution$, let $C_{f_i}$ denote the subset of clients from $R$ that are assigned to $f_i$. Since we take a uniform sample from each ring, in expectation we maintain the relative proportions of $|C_{f_i}|$'s in the sample. Therefore, in expectation, the sample maintains the cost as well as capacity constraints. However, it is quite likely that we \emph{over-} or \emph{under-sample} clients from different $C_{f_i}$'s, and we need to bound the cost of a feasible assignment of the weighted sample. To this end, the authors~\cite{Cohen-AddadL19Capacitated} use an instance of Minimum Cost Flow (\mcf) problem. They use the ring-center $f$ to reroute flows between different over- and under-sampled $C_{f_i}$'s, and the cost of rerouting can be bounded using a generalization of the above argument involving triangle inequality.

\medskip\noindent\textbf{\textsf{Our sampling and analysis.}} Since we need to handle both capacities and outliers, we need to be even more careful in our analysis. In a solution for \ckmo, unlike \ckm, not all clients in a ring $R$ may be assigned to facilities -- some can be outliers. 
In presence of capacities, this issue cannot be handled by simply ignoring the contribution of $\outliers$ farthest clients from $\solution$, unlike in \kmo. The optimal choice of weighted outliers from $R$, in fact, depends on the outcome of sampling outside $R$. Though, in expectation, we maintain relative proportions of $|C_{f_i}|$'s and of outliers in $R$, we may over- or under -sample clients from different $C_{f_i}$'s and over- or under -sample clients from the set of outliers.  However, we are able to decouple the analysis of under/over-sampling in each $C_{f_i}$ from the analysis of under/over-sampling in the set of outliers in $R$. At this point, the error incurred by sampling from the subset of assigned clients can be bounded by essentially following the argument from \cite{Cohen-AddadL19Capacitated}. Next, we bound the error incurred while sampling from the set of outliers. Since we treat each ring independently in the analysis, the subset of outliers from $R$ is essentially arbitrary.  To this end, we generalize the argument from Agrawal et al.~\cite{agrawal2023kMO}, and prove that rerouting cost in case of over/under sampling from the set of outliers is negligible for \emph{any} set of weighted clients of weight at most $m$, not just for those that are farthest from $\solution$. To bound the costs of rerouting, we analyze properties of a new problem, called Minimum Cost Flow with Outliers (\mcfo), which generalizes \mcf. The problem is used at several places in our paper and is of independent interest.

In conclusion, we show that the ring sampling approach can be used to construct a weighted sample $W$ of size $O\lr{\frac{(km \log n)^2}{\epsilon^3}}$ that approximately preserves the \ckmo cost w.r.t. any feasible solution $\solution \subseteq \facilities$ of size at most $k$, with high probability (\Cref{lem:coreset-multiring}).

Next, we discuss how to use $W$ to reduce \ckmo to \ckm. We enumerate all possible subsets of $\wtset' \subseteq \wtset$ of total \emph{weight} $m$ that corresponds to the set of outliers corresponding to an optimal solution $(F^*, \sigma^*)$ for the given instance. Using the bound on $|W|$, it can be easily argued that the number of such guesses is bounded by $f(k, m, \epsilon) \cdot |\cI|^{O(1)}$. For each such guess $W' \subseteq W$ for outliers, we obtain an instance of weighted \ckm by deleting $W'$, and it can be easily converted to unweighted \ckm by creating multiple co-located copies of each weighted client. For each such instance of \ckm, we use a $\gamma$-approximation, and we return the minimum-cost solution found over all guesses (for computing an optimal assignment for \ckmo, we need to be able to solve \mcfo in polynomial time, which we show in \Cref{lem:MCFO} by reducing it to \mcf). Using \Cref{lem:coreset-multiring}, we argue in \Cref{lem:finalcost} that in the iteration corresponding to the ``correct'' guess of $W'$, a $\gamma$-approximate solution for the \ckm instance is, in fact, a $(\gamma+\epsilon)$-approximation for the original instance of \ckmo.


\medskip\noindent\textbf{\textsf{Balanced Fair Clustering with Outliers.}} \textsc{$(\alpha, \beta)$-Fair $k$-Median} problem \textsf{($(\alpha, \beta)$-F$k$M)}, has been recently introduced in \cite{Chierichetti0LV17,Rosner018} and received significant attention. The model is motivated from the doctrine of disparate impact (DI) \cite{Chierichetti0LV17}, which states that protected attributes of population (such as sensitive demographic information) should not be used in decision making, and the decisions should not affect different demographic groups in a disproportionate manner. In the context of clustering, the protected attributes of population are modeled by dividing the clients into $\groups$ different groups, and the goal is to find a minimum-cost clustering that ensures the specified proportions of colors in each cluster. A particular case of interest is when we are required to find clusters, where the ``local proportions'' of each group are roughly equal to their ``global proportions''. Recently, Bandyapadhyay et al.~\cite{BandyapadhyayFS21} used the ring sampling approach to design FPT approximations for this problem and its variants. We study a generalization of this problem, called \textsc{$(\alpha, \beta)$-Fair $k$-Median with Outliers} (\textsf{$(\alpha, \beta)$-F$k$MO}), where we want to find a balanced clustering of clients, while excluding at most $m_{j}$ outliers of each group $1 \le j \le \groups$. Similar to \Cref{thm:kmed-approx-preserving}, we give a reduction from \textsf{$(\alpha, \beta)$-F$k$MO} to \textsf{$(\alpha, \beta)$-F$k$M}, which is the second main result of our paper. This is stated in the following informal theorem.

\begin{theorem}[Informal] \label{thm:fairkmed-approx-informal}
There exists an approximation-preserving reduction from $(\fairone, \fairtwo)$-\textsf{F$k$MO} to $(\fairone, \fairtwo)$-\textsf{F$k$M} that runs in time FPT in $k, m = \sum_{j \in [\groups]} m_j$, and $\epsilon$, where the underlying metric space remains unchanged.
\end{theorem}
Similar to \Cref{thm:kmed-approx-preserving}, this theorem can also be extended to the related \kmeans and \textsc{$k$-Facility Location} objectives. Some of the corollaries of this theorem are stated in \Cref{fig:corollaries}.

%% file: prelims.tex
\section{Preliminaries} \label{sec:prelims}
We will use $\mathbb{N}$, $\mathbb{R}$ and, $\mathbb{R}_{+}$ to denote the set of non-negative integers, the set of reals and the set of non-negative reals. 
Let ($\points$, $\dist$) be a metric where, $\points$ is a finite sets of points and $\dist:\points \times \points \rightarrow \mathbb{R}$ is a distance function satisfying triangle inequality and symmetry.

\medskip\noindent\textbf{\textsf{Minimum Cost Flow with Outliers (\mcfo)}.} In Minimum Cost Flow with Outliers problem (\mcfo), we are given a flow network $G=(V,E)$, set of facilities $\solution \subseteq V$, a set of demand points $\clients \subseteq V$ and an integer $\outliers$. Each facility $i \in \solution$ has some supply $\capacity{i}$, each demand point $j \in \clients$ has some demand $\wt(j)$ and every edge $(u,v) \in E$ has cost $\dist(u,v)$. The goal is to send at least $\sum_{j \in \clients} \wt(j) - \outliers$ demand from demand points in $\clients$ to facilities in $\solution$ using edges in $E$ such that every facility $i \in \solution$ receives flow less than or equal to its supply $\capacity{i}$ and the total cost is minimized (the total cost is the sum of the cost of the edges used to send the demand). \mcfo is a generalization of popular Minimum Cost Flow problem (\mcf) and is same as \mcf if $\outliers=0$. \mcf can be solved using linear programming in polynomial time.

\begin{restatable}{lemma}{mcfolemma}
\label{lem:MCFO}
There exists a polynomial time algorithm that returns an integral solution for \mcfo problem. 
\end{restatable}
\begin{proof}
	To solve \mcfo in polynomial time we will use a solution of \mcf. Given an instance $\I_{\mcfo}$ of \mcfo, we first create an instance $I'$ of \mcf as follows: create a dummy facility $\dummy$ with supply $\capacity{\dummy}=\outliers$. Add an edge from $\dummy$ to every other vertex $v \in V$ such that $\dist(\dummy,v) = D$ where $D$ is a value slightly greater than the largest distance in the \mcfo instance. Every demand point $j \in \clients$ is a demand point with total demand $\wt(j)$. Every facility $i \in \solution \cup \{ \dummy\}$ is a supply point with $\capacity{i}$ supply. 
	
	\begin{claim}   
		Let $\opt(\I_{\mcfo})$ be the cost of optimal solution to $\I_{\mcfo}$. $\opt(\I_{\mcfo}) = \opt(\I_{\mcf}) - X$ where $X$ is a fixed cost for a given instance $\I_{\mcfo}$.
	\end{claim}
	\begin{proof}
		Let $\outliers'= \max \{0, \sum_{j \in \clients}\wt(j) - \sum_{i\in \solution}\capacity{i}$ \}.
		
		We will prove this in two parts: $(i)$ $\opt(\I_{\mcfo}) \leq \opt(\I_{\mcf}) - X$ and $(ii)$ $\opt(\I_{\mcfo}) \geq \opt(\I_{\mcf}) - X$.
		
		To prove $\opt(\I_{\mcfo}) \leq \opt(\I_{\mcf}) - X$, let $S'$ be the optimal solution for \mcf with value $\opt(\I_{\mcf})$. A solution $S$ for \mcfo can be obtained from $S'$ as follows: discard the total flow coming onto the dummy facility $\dummy$ as outlier flow. Note that this demand is at most $\outliers$. If the flow coming onto $\dummy$ is strictly less than $\outliers$ (note that it must be $\outliers'$ in this case), we will make additional $\outliers-\outliers'$ demand outlier from other facilities. This is the most expensive $\outliers-\outliers'$ demand that is served in $S'$ from facilities in $\solution$. Let $Cost(\outliers-\outliers')$ represents the cost of this demand in the flow $S'$.
		We have, $\opt(\I_{\mcfo}) \leq Cost(S) = \opt(\I_{\mcf}) - \outliers' \cdot  \largedist - Cost(\outliers-\outliers')$. The claim follows for $X= \outliers' \cdot  \largedist + Cost(\outliers-\outliers')$.
		
		Next to prove $\opt(\I_{\mcfo}) \geq \opt(\I_{\mcf}) - X$ or $\opt(\I_{\mcf}) \leq \opt(\I_{\mcf}) + X$, we will create a feasible solution $S'$ of $\I_{\mcf}$ from the optimal solution $S$ of $\I_{\mcfo}$ as follows: send $\outliers'$ outlier demand to $\dummy$ at $\outliers' \cdot  \largedist$ cost. Send remaining $\outliers-\outliers'$ demand to facilities in $\solution$ using \mcf. Note that, at least $\outliers-\outliers'$ supply is available at facilities in $\solution$ by definition of $\outliers'$. Let $Cost'(\outliers-\outliers')$ be the cost paid for this. If $Cost'(\outliers-\outliers') \leq Cost(\outliers-\outliers')$, we are done. $\opt(\I_{\mcf}) \leq Cost(S') \leq \opt(\I_{\mcfo}) + X$ for $X= \outliers' \cdot  \largedist + Cost(\outliers-\outliers')$. Otherwise, there exist a solution to $\I_{\mcfo}$ of cost $< \opt(\I_{\mcfo})$ which is a contradiction.
	\end{proof}
\end{proof}

Given an instance $\cI$ of \ckmo and a set of facilities $\solution \subseteq \facilities$, we define $\cost_{\outliers}(\clients,\solution)$ as the cost of optimal assignment with $\outliers$ outliers. If the sum of the capacities of the facilities in $\solution$ is less than $n - \outliers$, then we call such a set $\solution$ to be infeasible and define  $\cost_{\outliers}(\clients,\solution) = \infty$.
Note that, for a given feasible $\solution$, the set $\clients'$ of outliers and the assignment of clients to facilities can easily be determined by solving a minimum cost flow with outliers problem (\mcfo). MCFO can be solved in polynomial time to return an integral solution as stated in \Cref{lem:MCFO}. Hence, $\cost_{\outliers}(\clients,\solution)$ can be computed for a given $\solution \subseteq \facilities$.

In this paper, we will often have non-negative integer weights on clients. Essentially, the weight on a client can be thought of as multiple co-located copies of a client. A formal definition follows. 
\begin{definition} (\wckmo).\label{def:wckmo}
The input is $\cI= ((\points,\dist),\wtset,\facilities,k, \capacity{}, \outliers)$ where $\wtset \subseteq \clients \times \nat$ is the set of pairs $\{(c,\wt(c)):c\in \clients\}$ and $\wt : \clients \rightarrow \nat$ is a weight function. The objective now is to find a subset $\solution \subseteq \facilities$ of size at most $k$ and an assignment $\assign : \clients \times \solution \rightarrow \nat$ such that:
\begin{enumerate}
    \item For each $c \in \clients$, $\sum_{f \in \solution} \assign(c,f) \leq \wt(c)$,
    \item Total unassigned weight is at most $\outliers$, i.e., $\sum_{c \in \clients} \wt(c) - \sum_{c \in \clients, f \in \solution} \assign(c,f) \leq \outliers$, 
    \item Assignment must respect capacity constraints for each $f \in \solution$, i.e., $\displaystyle\sum_{c \in \clients} \assign(c,f) \leq \capacity{f}$, and
    \item $\sum_{c \in \clients, f \in \solution} \assign(c,f) \dist(c,f)$ is minimized.
\end{enumerate}
\end{definition}

Similar to \ckmo, given an instance of \wckmo and a set of facilities $\solution \subseteq \facilities$, we define $\wcost_{\outliers}(\wtset,\solution)$ as the cost of optimal assignment where the total unassigned weight is at most $\outliers$. If the sum of the capacities of the facilities in $\solution$ is less than $\sum_{c \in \clients}\wt(c) - \outliers$, then we call such a set $\solution$ to be infeasible and define  $\wcost_{\outliers}(\clients,\solution) = \infty$.
Just like \ckmo, for a given feasible $\solution$, the assignment of clients to facilities for \wckmo can also be easily determined by solving \mcfo. Hence, $\wcost_{\outliers}(\wtset,\solution)$ can be obtained for a given $\solution \subseteq \facilities$.

%% file: ckmo.tex
\section{The Algorithm}
Given an instance $\cI=((\points,\dist),\clients,\facilities,k, \capacity{},\outliers)$ of \ckmo, we first create an instance $\IwO=((\points,\dist),\clients,\facilities \cup \clients,k+\outliers)$ of ($k+\outliers$)-Median (un-capacitated) where we have a facility co-located with every client in addition to the original set of facilities. The following claim follows from the easy observation that any solution $\cI$ is a feasible solution for $\IwO$.

\begin{restatable}{claim}{bdopti}
\label{lem:bd-optI}
Value of an optimal solution to $\IwO$ is bounded by the value of an optimal solution to $\cI$, that is, $\opt(\IwO) \leq \opt(\cI)$.
\end{restatable} 

Now, the instance $\IwO$ is solved using any polynomial time algorithm (say \cite{GowdaPST23}) to obtain a constant ($\factortwo$) factor approximation. Let $\solution_{\factortwo}$ be the set of facilities opened by the algorithm. Note that $k_{\factortwo} = |\solution_{\factortwo}| \leq k+\outliers$ and 
\begin{equation}
  \label{eqn:alpha-opt-bd}  \cost_0(\clients,\solution_{\factortwo}) \leq \factortwo \cdot \opt(\IwO) \leq \factortwo \cdot \opt(\cI)
\end{equation} 
 For every facility $f \in \solution_{\factortwo}$, let $\clients_{f} \subseteq \clients$ be the clients assigned to $f$ by the algorithm. Note that the sets $\clients_{f}$ are disjoint. Let $\ball(f,X) \subseteq \clients_f$ denote the ball of radius $X$ consisting of clients, in $\clients_f$, within distance $X$ from $f$. Let $\rad = \frac{\cost_0(\clients,\solution_{\factortwo})}{\factortwo n}$, and $\psi=\ceil{\log(\factortwo n)}$.

For a facility $f \in \solution_{\factortwo}$ and set $\clients_f$, we further partition $\clients_f$ into smaller sets called as {\em rings} such that points in each ring have similar distances to $f$, i.e.,

\begin{equation}
 \clients_{f,j} =\begin{cases}
    \ball(f,\rad), & \text{if $j=0$}.\\
    \ball(f,2^j\rad) \setminus \ball(f,2^{j-1}\rad), & \text{if $1 \le j \le \psi$}.
  \end{cases} 
\end{equation}

Let $s= \frac{a\factortwo^2}{\epsilon^3}(\outliers + k \ln n)$ where $a$ is a large enough constant. We sample weighted points $\wtset^f$ from $\clients_f$ as follows: for every set $\clients_{f,j} \subseteq \clients_f$, if $|\clients_{f,j}| \leq s$ then for every client $c \in \clients_{f,j}$ add $(c,1)$ to $\wtset^f$. Else, sample $s$ random clients in $S_{f,j} \subseteq \clients_{f,j}$ without replacement. For every client $c \in S_{f,j}$, add $(c,\frac{|\clients_{f,j}|}{s})$ to $\wtset^f$.

Applying the above procedure on every cluster $f \in \solution_{\factortwo}$ we get,  $\wtset = \cup_{f \in \solution_{\factortwo}} \wtset^f$ as the weighted set of points. As the number of rings is $\log n$ and the number of points in each ring is at most $\frac{a\factortwo^2}{\epsilon^3}(\outliers + k \ln n)$, the total number of weighted points, $|\wtset| = O\lr{\frac{(km \log n)^2}{\epsilon^3}}$, since $a$ and $\factortwo$ are constants. Without loss of generality, we assume the weights are integral, by a slight modification of the construction above, which can increase the number of points in $\wtset$ by at most a factor of $2$. See, e.g., ~\cite{chen2009coresets,agrawal2023kMO} for details. 


Our key technical contribution is to show the following lemma (\Cref{lem:coreset-multiring}), which states that for all feasible sets $\solution \subseteq \facilities$, the cost of assignment of the weighted sample to $\solution$, after excluding outliers with total weight $\outliers$, is close to the cost of assigning original set of clients, after excluding $\outliers$ outliers to $\solution$. This is precisely the definition of $\epsilon$-coresets, which have been extensively studied, especially in the context of (uncapacitated) \kmed and \kmeans (see \cite{Cohen-AddadSS21} and references therein). Thus, \Cref{lem:coreset-multiring} implies the existence of small-sized coreset that handles both capacities and outlier constraints; however unlike the coreset literature, in this work our focus is not on optimizing the size of the coreset. The proof of the lemma is deferred to \Cref{sec:singlering} and \Cref{sec:multiring-full} for ease of understanding.

\begin{lemma}
\label{lem:coreset-multiring}
For all feasible sets $\solution \subseteq \facilities$ of size $k$, $|\wcost_{\outliers}(\wtset,\solution)-\cost_{\outliers}(\clients,\solution)| \leq \epsilon \cost_m(C, F)$ with probability at least $1-1/n$.
\end{lemma}

Assuming \Cref{lem:coreset-multiring} holds, let us return to our algorithm to see how it can be used to reduce \ckmo to \ckm. Recall that $\wtset$ is the set of weighted clients returned by our sampling algorithm and let $\wtpoints$ be the set of clients (i.e., the set of first elements from each pair $(c, \wt(c)) \in \wtset$, denoting a weighted client). For any subset $\wtout \subseteq \wtpoints$, let $\wt(\wtout) = \sum_{c_i \in \wtout} \wt(c_i)$.
The algorithm then proceeds as follows: we iterate over each guess $\wtout \subseteq \wtpoints$ of size at most $\outliers$. First, we check whether $\wt(\wtout) \ge \outliers$ -- if not, we continue to the next guess. Now, suppose $\wt(\wtout) \ge \outliers$. Then, let us order the points in $\wtout$ as $c_1, c_2, \ldots, c_{\outliers'}$, where $\outliers' = |\wtout|$. For each $c_i \in \wtout$, we guess an integer $0 < z(c_i) \le \min\LR{\outliers, \wt(c_i)}$, such that $\sum_{c_i \in \wtout} z(c_i) = \outliers$. Note that the number of guesses is at most $(\outliers')^{\outliers} \le \outliers^{\outliers}$. Now, fix one such guess $\vecz = (z(c_1), z(c_2), \ldots, z(c_\outliers'))$. Now, we define a new weight function $w_{T, \vecz}: \wtpoints \to \nat$ as follows. $$\wt_{T, \vecz}(c) = \begin{cases} \wt(c) &\text{ if } c \not\in \wtout
\\\wt(c_i) - z(c_i) &\text{ if } c = c_i \in \wtout
\end{cases}$$
Let $\wtset_{T, \vecz} = \LR{ (c, \wt_{T, \vecz}(c)) : c \in \wtpoints }$ denote the resulting set of weighted points, where the set of points is the same as that in $\wtset$, but the weight function is $\wt_{T, \vecz}$ instead of $\wt$. Now, the algorithm uses a $\factor$-approximation algorithm for the \wckm\footnote{Existing $\factor$-approximation algorithms for unweighted \ckm can be used to solve weighted instances by creating multiple co-located copies of each weighted client.} instance $\IwO = ((\points,\dist),\wtset_{T, \vecz},\facilities,k, \outliers_0 = 0)$ to find a $\solution \subseteq \facilities$ of size at most $k$ and an assignment $\assign : \clients \times \solution \rightarrow \nat$ satisfying the first three properties in \Cref{def:wckmo}. Then, we use \mcfo to compute the optimal assignment cost from $\clients$ to $\solution$ with $\outliers$ outliers, and the corresponding assignment $\assign$. Finally, the algorithm returns a solution $(\solution^*, \assign^*)$ of the minimum cost, over all guesses (i.e., guesses for $\wtout$ as well as $\vecz$). First, in \Cref{lem:runtime} we analyze the running time of the algorithm, and then in \Cref{lem:finalcost}, we analyze the approximation guarantee.

\begin{restatable}{lemma}{runtime} \label{lem:runtime}
    Let $T(|I'|, k')$ denote the running time of the $\factor$-approximation used to solve an instance $I'$ of \textsc{Capacitated $k'$-Median}. Then, the running time of our algorithm is upper bounded by $f(k, m, \epsilon) \cdot T(|\cI|, k) \cdot |\cI|^{O(1)}$.
\end{restatable}
\begin{proof}
	Note that construction of the set $W$ takes polynomial time, since we use $\factortwo$-approximation for \textsc{$k+m$-Median} as a starting point.
	Next, we bound the total number of guesses tried by the algorithm. First, we guess a subset $\wtout \subseteq \wtset$ of size at most $\outliers$, which takes time $\binom{|\wtset|}{m} \le \lr{\frac{km \log n}{\epsilon}}^{O(m)}$, which can be upper bounded by $\lr{\frac{km}{\epsilon}}^{O(m)} \cdot n^{O(1)}$ via a standard case analysis on whether $m \le \frac{\log n}{\log \log n}$. Next, for each such guess $\wtout$, the algorithm guesses the vector $\vecz = (z(p_1), z(p_2), \ldots, z(p_{\outliers'}))$. As remarked earlier, the number of such guesses is upper bounded by $\outliers^{\outliers}$. Finally, for each such guess, we use a $\gamma$-approximation for \ckm, which takes time $T(|\cI|, k)$, and for the solution $\solution \subseteq \facilities$ returned by the algorithm, we use \mcfo to compute the cost w.r.t. the original set $\clients$ which runs in polynomial time due to \Cref{lem:MCFO}. Thus, the claimed bound on the running time follows.
\end{proof}

Now, we analyze the cost of this solution.

\begin{restatable}{lemma}{finalcost} \label{lem:finalcost}
    Let $\solution^* \subseteq \facilities$ be the set of at most $k$ facilities returned by the algorithm. Then, with probability at least $1-1/n$, it holds that for any feasible set $\solution \subseteq \facilities$ of size at most $k$, 
    $$\cost_{\outliers}(\clients, \solution^*) \le \factor \cdot \frac{1+\epsilon}{1-\epsilon} \cdot \cost_{\outliers}(\clients, \solution).$$
\end{restatable}
\begin{proof}
	The statement in \Cref{lem:coreset-multiring} holds with probability at least $1-1/n$. We condition on this good event, and show that the current lemma holds (with probability $1$). Fix a feasible set $\solution \subseteq \facilities$ of size at most $k$, as in the statement of the lemma. Then, by \Cref{lem:coreset-multiring}, we have the following inequality:
	\begin{equation}
		(1-\epsilon) \cdot \cost_{\outliers}(\clients, \solution) \le \wcost_{\outliers}(\wtset, \solution) \le (1+\epsilon) \cdot \cost_{\outliers}(\clients, \solution). \label{eqn:coresetbound1}
	\end{equation}
	Now, consider the assignment $\sigma$ realizing $\wcost_{\outliers}(\wtset, \solution)$, and let $T \subseteq \wtpoints$ be the set of clients $c$, such that $\sum_{f \in F} \sigma(c, f) < w(c)$. Note that $z(c) \coloneqq w(c) - \sum_{f \in F} \sigma(c, f)$ is an integer for every $c \in \wtpoints$. Let $\vecz = (z(c_1), z(c_2), \ldots, z(c_{\outliers}'))$, where $T = \LR{c_1, c_2, \ldots, c_{\outliers'}}$ is indexed arbitrarily. 
	
	It is easy to verify that $(T, \vecz)$ as defined here, satisfy the conditions of a ``guess'' as in the algorithm. Thus, consider the iteration of the algorithm corresponding to $T$ and $\vecz$. It follows that $\wcost_{\outliers}(\wtset, \solution) = \wcost_0(\wtset_{T, \vecz}, \solution)$. Let $\solution' \subseteq \facilities$ be the $\factor$-approximation found for the \wckm instance $\IwO = ((\points,\dist),\wtset_{T, \vecz},\facilities,k, 0)$ in this iteration. It follows that,
	\begin{equation}
		\wcost_{\outliers}(\wtset, \solution') = \wcost_0(\wtset_{T, \vecz}, \solution') \le \factor \cdot \opt(\cI) \le \factor \cdot \wcost_0(\wtset_{T, \vecz}, \solution) \label{eqn:approxbound1}
	\end{equation}
	Here, $\opt(\cI)$ denotes the cost of an optimal solution for the \ckmo instance $\cI$, and since $\solution$ is a feasible set of at most $k$ facilities, the inequality $\opt(\cI) \le \wcost_0(\wtset_{T, \vecz}, \solution)$ follows. 
	
	Now, again from \Cref{lem:coreset-multiring}, we have the following inequality
	\begin{equation}
		(1-\epsilon) \cdot \cost_{\outliers}(\clients, \solution') \le \wcost_{\outliers}(\wtset, \solution') \le (1+\epsilon) \cdot \cost_{\outliers}(\clients, \solution') \label{eqn:coresetbound2}
	\end{equation}
	Finally, $\solution^* \subseteq \facilities$ is the best solution found over all iterations, it follows that
	\begin{align}
		\cost_{\outliers}(\clients, \solution^*) \le \cost_{\outliers}(\clients, \solution') &\le \frac{1}{1-\epsilon} \cdot \factor \cdot \wcost_0(\wtset_{T, \vecz}, \solution) \tag{From (\ref{eqn:approxbound1}) and (\ref{eqn:coresetbound2}}
		\\&=  \frac{1}{1-\epsilon} \cdot \factor \cdot \wcost_{\outliers}(\wtset, \solution) \tag{By definition of $T, \vecz$}
		\\&\le \factor \cdot \frac{1+\epsilon}{1-\epsilon} \cdot \cost_{\outliers}(\clients, \solution) \tag{From (\ref{eqn:coresetbound1})}.
	\end{align}
\end{proof}

\Cref{lem:runtime} and \Cref{lem:finalcost} together yield a proof of the following theorem, which is the formal version of \Cref{thm:kmed-approx-preserving}.

\begin{theorem}\label{thm:kmed-approx-preserving-formal}
    Given a $\factor$-approximation for \textsc{Capacitated $k'$-Median} that runs in time $T(|I'|, k')$ on an input instance $I'$, we can obtain a $(\gamma+\epsilon)$-approximation for an instance $\cI$ of \ckmo in time $f(k, m, \epsilon) \cdot T(|\cI|, k) \cdot |\cI|^{O(1)}$.
\end{theorem}

To complete the proof of \Cref{thm:kmed-approx-preserving-formal}, it remains to prove \Cref{lem:coreset-multiring}. First, in \Cref{sec:singlering}, we prove \Cref{lem:coreset-multiring} in a special case. Then, in \Cref{sec:multiring-full}, we show how this analysis can be extended to the more general case.

\section{Single Ring Case} \label{sec:singlering}
In this section, we will prove a version of \Cref{lem:coreset-multiring} in the following special case that already illustrates most of the technical details that will be later used to prove the more general case. 

Consider an arbitrary ring $\ringf{j}$ in a cluster $\clusterf$ centered at a facility $f \in \solution_{\factortwo}$ that will remain fixed throughout this section. We assume that the algorithm performs sampling only inside $\ringf{j}$; whereas all the clients in rings different from $\ringf{j}$ are all included into $\wtset$ with weight 1. Note that this case can occur in the actual algorithm, if, all rings other than $\ringf{j}$ contain fewer than $s$ points. Let $\radr=2^j\rad$ be the radius of this ring. Let $N=|\ring|$ be the number of clients in the ring. Let $\wtset = \cup_{f}\wtset^f$ be the weighted clients returned by our algorithm. In the remaining section we will prove the following lemma.

\begin{lemma}
\label{lem:coreset}
For any feasible $\solution \subseteq \facilities$ of size $k$, $|\wcost_{\outliers}(\wtset,\solution)-\cost_{\outliers}(\clients,\solution)| \leq \epsilon \constb N \radr$ with probability $1-n^{-(k+\consta)}$ where $\consta$ and $\constb$ are constants.
\end{lemma}


As stated in the lemma, we fix a feasible set $\solution \subseteq \facilities$. For the rest of the section we will assume $N > s = O\lr{\frac{km \log n}{\epsilon^3}}$, because otherwise the sampling does not change anything (recall that all clients outside $\ringf{j}$ already belong to $\wtset$ with weight $1$). For our analysis, we will define a random vector $\vectorx$ and a function $\func$ such that $\func(\vectorx)$ and $\wcost_\outliers(\wtset,\solution)$ are identically distributed. 

\textbf{Defining $\vectorx$:} a random vector $\vectorx \in \mathbb{R}_{+}^N$ is defined as follows: for each coordinate pick value $N/s$ with probability $s/N$ and 0 otherwise such that $\ex[\vectorx]=\vecone$. One can show that with sufficient probability, this sampling scheme
selects exactly $s$ points using the \Cref{claim:bin}.

\begin{claim}[\cite{Cohen-AddadL19Capacitated}]
\label{claim:bin}
   Let $b$ be a positive integer, and let $ p \in (0,1)$ such that $p b$ is an integer. The probability that Binomial($b,p$)= $pb$
 is $\Omega(1/\sqrt{b})$.
 \end{claim}

 Setting $b=N$ and $p=s/N$ in \Cref{claim:bin}, it follows that $\vectorx$ has exactly $N \cdot (s/N) = s$ non-zero entries with probability $\Omega(1/\sqrt{N})$. Conditioned on this event then, $\vectorx$ and $\wtset$ are identically distributed, i.e., the $\vectorx$ represents the outcome of our sampling algorithm. In the rest of the section, we analyze the unconditioned behavior of $\vectorx$, and show that the desired concentration (as in \Cref{lem:coreset}) holds with high probability. Then, a standard argument shows that \Cref{lem:coreset} also holds with high probability, even when we condition on the event that $\vectorx$ has exactly $s$ non-zero entries. 

\textbf{Defining Function $\func$:} to define $\func$, we first create an instance of \mcfo. Given a vector $\vecv$ of size $N$ where each entry in $\vecv$ corresponds to a client in $\ring$, a flow instance $\flow(\vecv)$ is created as follows: every client $c \in \ring$ has $\vecv_c$ units of demand, every client in $\clients\setminus\ring$ has $1$ unit of demand, cluster center $f$ has $N-\sum_{c \in \ring}\vecv_c$ (possibly negative\footnote{Negative demand $d$ at a vertex $v$ requires that $d$ units of flow must enter $v$, whereas a positive demand requires that the specified units of demand must exit the vertex.}) demand. Every facility $f' \in \solution$ has $\capacity{f'}$ units of supply. The number of outliers is $\outliers$. We define $\func(\vecv): \vecv \rightarrow \mathbb{R}_{+}$ as the cost of optimal flow of $\flow(\vecv)$. The flow instance we defined is feasible because the sum of demands is $|\clients|=n$ and the number of outliers is $\outliers$, making the demand to be served to be $n-\outliers$ which is feasible by assumption on $\solution$. 

Note that, $\func(\ex[\vectorx]) = \func(\mathbbm{1})$ is exactly $\cost_{\outliers}(\clients,\solution)$. Also, $\func(\vectorx)$ and $\wcost_\outliers(\wtset,\solution)$ are identically distributed. Now, we will prove \Cref{lem:coreset} in two steps: $(i)$ $\func(\vectorx) \approx \ex(\func(\vectorx))$ with high probability (proven in \Cref{lem:part1}) and $(ii)$ $\ex(\func(\vectorx)) \approx \func(\ex[\vectorx])$ (proven in \Cref{lem:part2}).

\begin{restatable}{lemma}{lempartone}
\label{lem:part1}
$|\func(\vectorx)-\ex[\func(\vectorx)]| \leq \epsilon N \radr /2$ with probability $\geq 1-n^{-(k+c)}$.
\end{restatable}
\begin{proof}
	We will first show that $\func(\vectorx)$ is $\radr$-Lipschitz with respect to the $\ell_1$ distance in $\mathbb{R}^N_{+}$ and then apply standard martingles tools to prove that $\func(\vectorx)$ is concentrated around its mean. To prove $\func(\vectorx)$ is $\radr$-Lipschitz, fix a client $c \in \ring$, and consider two vectors $\vecv$, $\vecv' \in \mathbb{R}^N_{+}$ with $\vecv'= \vecv + \delta \cdot \mathbbm{1}_{c}$. Note that the flow instance $\vecv$ is same as flow instance of $\vecv'$ except the later has $\delta$ more demand at client $c$ and $\delta$ less demand at cluster center $f$. 
	
	We will first construct a feasible flow for $FI(\vecv')$ from optimal flow $\fl$ of $FI(\vecv)$ as follows: create a dummy facility $\dummy$ with $\outliers$ units of supply and connect it to all the demand points by introducing edges with cost $0$. For every demand point, the amount of demand that was outlier in $\fl$ of $FI(\vecv)$ is sent to $\dummy$ at $0$ cost. To construct feasible flow for $\vecv'$, add $\delta$ units of flow from $c$ to $f$. Make all the demand coming on to $\dummy$ facility outlier. It is easy to see that the resulting flow is a feasible flow for $\vecv'$ and the cost of solution increases by at most $\delta \radr$. Therefore, $\func(\vecv') \leq \func(\vecv) + \delta \radr$.
	
	We next construct a feasible flow for $FI(\vecv)$ from optimal flow $\fl$ of $FI(\vecv')$ in a similar way: create a dummy facility $\dummy$ with $\outliers$ units of supply and connect it to all the demand points by introducing edges with cost $0$. For every demand point, the amount of demand that was outlier in $\fl$ of $FI(\vecv')$ is sent to $\dummy$ at $0$ cost. To construct feasible flow for $\vecv$, add $\delta$ units of flow from $f$ to $c$. Make all the demand coming on to $\dummy$ facility outlier. Again it is easy to see that this is a feasible flow for $\vecv$ and the cost of solution increases by at most $\delta \radr$. Therefore, $\func(\vecv) \leq \func(\vecv') + \delta \radr$. 
	
	The proof now follows in a similar way as done in Cohen-Addad and Li~\cite{Cohen-AddadL19Capacitated}, that is, the desired bound is obtained by applying the Chernoff bound for Lipschitz functions (stated in the following \Cref{lem:lip-func}).
	
	\begin{proposition}[\cite{Cohen-AddadL19Capacitated}]
		\label{lem:lip-func}
		Let $x_1,\ldots,x_n$ be independent random variables taking value $b$ with probability $p$ and value $0$ with probability $1-p$, and let $\func : [0,1]^n \rightarrow \mathbb{R}$ be a $L$-Lipschitz function in $\ell_1$ norm. Define $X:=(x_1,\ldots,x_n)$ and $\mu := \ex[\func(X)]$. Then, for $0 \leq \epsilon \leq 1$, $Pr[|\func(X)-\ex[\func(X)|] \geq \epsilon p n b L] \leq 2e^{-\epsilon^2pn/3}.$ 
	\end{proposition}
	
	We apply \Cref{lem:lip-func} on function $\func$ with $X := \vectorx$, $p := s/N$, $n := N$, $b:= 1/p$, and $L := \radr$ to obtain the following: 
	
	\begin{equation}
		\begin{split}
			\Pr[|\func(\vectorx)-\ex[\func(\vectorx)]| \geq (\epsilon/2)N \radr] & = \Pr[|\func(\vectorx)-\ex[\func(\vectorx)|] \geq (\epsilon/2)pnbL]  \\ & \leq 2 \cdot \exp \lr{\frac{-(\epsilon/2)^2pn}{3}}  \\ & = 2 \cdot \exp \lr{\frac{-(\epsilon/2)^2(s/N)N}{3}} \\ & = \exp \lr{-\Theta\lr{\epsilon^2 s}} \\ & = \exp\lr{-\Theta\lr{\epsilon^2 \cdot \frac{\outliers+k \log n}{\epsilon^2}}} \\ & \leq n^{-(k+\consta)} 
		\end{split}
	\end{equation}
	where the last equality follows by definition of $s$ and $\consta$ is a constant in last inequality. This concludes proof of \Cref{lem:part1}.
\end{proof}

\begin{lemma}
\label{lem:part2}
$|\ex[\func(\vectorx)]-\func(\ex[\vectorx])| \leq \epsilon \constb N \radr$ where $\constb$ is a constant.
\end{lemma}

\begin{proof}
We prove this in two parts $(i)$ $\func(\ex[\vectorx]) \leq \ex[\func(\vectorx)]$ (proven in \Cref{sec:pf-lem1}), and $(ii)$ $\ex[\func(\vectorx)] \leq \func(\ex[\vectorx]) + \epsilon \constb N  \radr$ (proven in \Cref{sec:pf-lem2}).
\end{proof}

\Cref{lem:coreset} follows by adding \Cref{lem:part1} and \Cref{lem:part2} and modifying $\constb\gets\constb+1/2$.

\subsection{Proof of \Cref{lem:part2}, part \texorpdfstring{$(i)$ $\func(\ex[\vectorx]) \leq \ex[\func(\vectorx)]$}{(i)}}
\label{sec:pf-lem1}
To prove, $\func(\ex[\vectorx]) \leq \ex[\func(\vectorx)]$, we construct a feasible solution for $FI(\ex[\vectorx])$ of cost no more than $\ex[\func(\vectorx)]$. Since the min-cost flow can only be lower, we get the desired result.

Let the outcomes of vector $\vectorx$ be $\vecv_1, \vecv_2, \ldots$ with probability $p_1,p_2,\ldots$ respectively. We have, $\ex[\func(\vectorx)] = \sum_{ i } p_i \func(\vecv_i)$. Let $\fl_i$ be the flow obtained for $\flow(\vecv_i)$. Now, consider the flow $\fl$ obtained by summing up over $i$, $\fl_i$ scaled by $p_i$. Observe that the cost of $\fl$ is at most $\sum_{ i } p_i \func(\vecv_i)$ which is $=\ex[\func(\vectorx)]$. 

Next, we will show that $\fl$ is a feasible flow for $\flow(\sum_i p_i \vecv_i) = \flow(\ex(\vectorx))$.
For a client $c$, let $y_{c,i}$ be the demand in $\flow(\vecv_i)$ and let $o_{c,i}$ be the demand that is left as outlier, i.e., $(y_{c,i}-o_{c,i})$ demand is satisfied  in $\fl_i$.
Therefore, in $\fl$, total $\sum_{c \in \clients}\sum_{i} p_{i} \cdot (y_{c,i}-o_{c,i})$ demand is satisfied. And, $\sum_{c \in \clients}\sum_{i} p_{i} \cdot (y_{c,i}-o_{c,i}) = \sum_{i} p_{i} \cdot \sum_{c \in \clients} (y_{c,i}-o_{c,i}) = \sum_{c \in \clients} (y_{c,i}-o_{c,i}) = |\clients| - \outliers$. The second last equality follows because the sum of probabilities is $1$ and last equality follows as $\phi_i$ is a feasible flow of $\flow(\vecv_i)$. 


Next we show that the capacities are respected on every facility in flow $\fl$. For any facility $j \in \solution$, let $j_i$ be the total flow coming onto $j$ in $\fl_i$. Therefore, total flow coming on to $j$ in $\fl$ is $\sum_{i}j_i p_i \leq (\sum_i p_i) \cdot \max_{i} j_i = \max_{i} j_i = j_{i^*} $ (say) which is at most the capacity of $j$ as $\fl^*$ is a feasible flow of $\flow(\vecv_i)$.

\subsection{Proof of \Cref{lem:part2}, part \texorpdfstring{$(ii)$ $\ex[\func(\vectorx)] \leq \func(\ex[\vectorx]) + \epsilon \constb N  \radr$}{(ii)}}
\label{sec:pf-lem2}
    
To prove this we first prove  \Cref{lem:gx1} followed by \Cref{lem:gx2}.


\begin{restatable}{claim}{loosebound} \label{lem:gx1}
With probability 1, $\func(\vectorx) \leq \func(\ex[\vectorx]) + n N \radr$.
\end{restatable}
\begin{proof}
	The value of $\func(\vectorx)$ lies in an interval of length $N/s \cdot N \cdot \radr \leq N^2 \radr$ because $\vectorx \in [0, N/s]^N$ and function $\func$ is $\radr$-Lipschitz. $\func(\ex[\vectorx])$ also lies in the same interval because $\ex[\vectorx] = \mathbbm{1} \in [0, N/s]^N$. Therefore, the claim follows.
\end{proof}
Now we prove the following lemma, which is the key technical step in our analysis.
\begin{lemma}
\label{lem:gx2}
    With probability $1-n^{-10}$, $\func(\vectorx) \leq \func(\ex[\vectorx]) +  \epsilon \constb N \radr$ where $\constb$ is a constant.
\end{lemma}

\begin{proof}
    To prove this, we construct a feasible flow $\fl$ for $FI(\vectorx)$ from the min cost flow with outliers of $FI(\mathbbm{1})$ such that cost of flow $\fl$ is bounded by $\func(\ex[\vectorx]) + \epsilon \constb N \radr$.
    
    Consider the optimal min-cost flow with outliers $\flowS$ of $\flow(\mathbbm{1})$. Let $\out$ be the set of clients that are made outliers in this flow. Note that, $|\out| \leq \outliers$. 
    Create a dummy facility $\dummy$ and connect it to all the clients in $\out$ at $0$ cost. For every client $c \in \out$, add 1 unit of flow from $c$ to $\dummy$ in $\flowS$. Note that, $\func(\mathbbm{1})$ does not change. Also, set the cost from ring center $f'$ to $\dummy$ to $0$. 
    
    For all the clients that are not in $\ring$, we route their 1 unit of demand in the same way as in $\flowS$. We are left with demands in $\ring$ and the extra $N - \sum_{c \in \ring} \vecv_c$ demand at the ring center $f'$. 
    
    For every facility $f \in \facilities \cup \dummy$, let $\ring^{f} \subseteq \ring $ be the set of clients served by $f$ in $\flowS$ and let $\samplef \subseteq \ring^f$ be the sampled clients in $\ring^f$. If $\samplef$ is under-sampled, i.e., $\samplef \cdot \frac{N}{s} \leq \ring^{f}$, then in $\fl$ we route :
    \begin{enumerate}
        \item $\frac{N}{s}$ units of flow from every $c \in \samplef$ to facility $f$ and,
        \item send $|\ring^f| - |\samplef| \cdot \frac{N}{s}$ units of flow from ring center $f'$ to facility $f$.
    \end{enumerate}

    Whereas if $\samplef$ is over-sampled, i.e., $|\samplef| \cdot \frac{N}{s} > |\ring^{f}|$, then we first pick a sub sample randomly from $\samplef$, say $\samplef_s$ of size $\lfloor |\ring^f| \cdot \frac{s}{N} \rfloor$. Now in $\fl$ we route:
    \begin{enumerate}
        \item $\frac{N}{s}$ units of flow from every $c \in \samplef_s$ to facility $f$,
        \item send $|\ring^f| - |\samplef_s| \cdot \frac{N}{s}$ units of flow from ring center $f'$ to facility $f$ and,
        \item $\frac{N}{s}$ units of flow from every $c \in \samplef \setminus \samplef_s $ to ring center $f'$.
    \end{enumerate}

Observe that the total amount of incoming flow on $\dummy$ is $|\out| \leq \outliers$. Make all the demand coming on to $\dummy$ outlier in $\fl$. It can be easily verified that the resulting flow $\fl$ is a feasible flow for instance $\flow(\vectorx)$.  

The cost of flow of the clients in $\ring \setminus \out$ is same as in Cohen-Addad and Li~\cite{Cohen-AddadL19Capacitated} as stated in following lemma:

\begin{lemma}[\cite{Cohen-AddadL19Capacitated}]
\label{ineq_cost_bd}
    With probability at least $1-n^{-10}$, the cost of clients in $\ring \setminus \out$ in flow $\fl$ is bounded by  $\displaystyle \sum_{f \in \solution \setminus \dummy} \sum_{c \in \ring^f \setminus \out} \dist(c,f) + 0.48 \epsilon N \radr$.
\end{lemma}

Now, for clients in $\out$, the cost paid from client to $\dummy$ or from ring center to $\dummy$ is $0$. Therefore, the only additional cost we paid is from a client $c \in \out$ to the ring center $f'$ which is at most $\radr$. Since $|\out| \leq \outliers$, the total additional cost is at most $\outliers \cdot \radr \leq \frac{\epsilon}{\constb} N \cdot \radr$ where $\constb$ is a constant. Adding this to Lemma~\ref{ineq_cost_bd} we obtain a total cost of at most,

\begin{equation}
    \sum_{f \in \solution}\sum_{c \in \ring} \dist(c,f) + 0.48 \epsilon N \radr + \epsilon \constb N\radr. 
\end{equation}

Modifying constant $\constb= 0.48 + \constb$ gives us \Cref{lem:gx2}.
\end{proof}

From \Cref{lem:gx1} and \Cref{lem:gx2}, we have
\begin{equation}
    \begin{split}
      \ex[\func(\vectorx)] & \leq n^{-10} \cdot (\func(\ex[\vectorx]) +n N \radr) + (1-n^{-10}) (\func(\ex[\vectorx]) + \constb \epsilon N \radr)  \\ & = \func(\ex[\vectorx]) + (n^{-10}\cdot n + (1-n^{-10})\cdot \constb)N \radr \\ & \leq \func(\ex[\vectorx]) + \epsilon \constb N \radr,  
    \end{split}
\end{equation}
and this completes the proof that $\ex[\func(\vectorx)] \leq \func(\ex[\vectorx]) + \epsilon \constb N  \radr$.

\input{multiring}

\section{Extensions to Other Objectives and \texorpdfstring{$(\alpha, \beta)$-}{}Fairness} \label{sec:extensions-short}
\subsection{Capacitated \texorpdfstring{$k$-Means/$k$-Facility Location}{k-Means/k-Facility Location} with Outliers} \label{subsec:kmean-fl}
First we consider \textsc{Capacitated $k$-Facility Location with Outliers} (\textsf{C$k$FLO}), which is a generalization of \ckmo with facility opening costs. Here, each $f \in \facilities$ has an \emph{opening cost} $o_f \ge 0$ (the opening costs may be different). \textsf{C$k$FLO} is defined exactly as \ckmo, except that the objective cost now also involves the opening costs. More formally, we want to find sets $\clients' \subseteq \clients$ and $\solution \subseteq \facilities$ that satisfies the other conditions, and minimizes $ \sum_{f \in \solution} o_f + \sum_{c \in \clients \setminus \clients'} \dist(c, \sigma(s)).$
We first observe that the $3+\epsilon$-approximation from \cite{Cohen-AddadL19Capacitated} can be easily adapted for \textsf{C$k$FLO} -- in the enumeration part, one also has to guess the opening cost of the closest facility to each ``leader'' up to a factor of $1+\epsilon$. Then, our reduction from \ckmo to \ckm also works in presence of opening costs -- indeed, the weighted sample $W$ approximately preserves distances w.r.t. all sets $\solution \subseteq \facilities$ w.h.p., and the opening costs of the facilities are unaffected by the sampling process. Thus, \Cref{thm:kmed-approx-preserving}, which gives a reduction from \ckmo to \ckm, in fact generalizes to give an FPT reduction from \textsf{C$k$FLO} to \textsf{C$k$FL}. In particular, we obtain $(3+\epsilon)$-approximation for \textsf{C$k$FLO} in general metrics.

\textsc{Capacitated $(k, z)$-Clustering with Outliers} is defined similar to \ckmo, except that the objective function now involves $z$-th power of distances, where $z \ge 1$ is fixed. In particular, $z = 1$ corresponds to \ckmo and $z = 2$ corresponds to \textsc{Capacitated $k$-Means with Outliers}. The ring sampling approach is known to generalize to handle the $z$-th powers of distances in the objective, however the analysis becomes more tedious, since the $z$-th power of distances only satisfies a ``weaker version'' of triangle inequality. However, we have to be careful in a few steps. For example, the guarantee in \Cref{lem:gx2} while rerouting flow of clients connected to different facilities, as well as that of outliers, we have to incur an additional $O_z(\epsilon) \cdot \opt$ term for each ring, where the constant in the big-Oh notation hides factors exponential in $z$ (which is acceptable since $z$ is fixed). When we add the errors over all rings, we now have an additional $O_z((k+m)\log n\cdot \epsilon) \cdot \opt$ term. Thus, by appropriately rescaling $\epsilon$, we obtain the desired guarantee. We leave the technical details for a future version.

\subsection{\texorpdfstring{($\fairone,\fairtwo$)-}{}Fair Clustering with Outliers}
Here, we give a brief overview of our reduction from $(\fairone, \fairtwo)$-\textsc{Fair $k$-Median with Outliers} (\textsf{$(\fairone, \fairtwo)$-BF$k$MO}) to its outlier-free counterpart. A formal definition of  \textsf{$(\fairone, \fairtwo)$-F$k$MO} follows.

\begin{definition}[\textsf{$(\fairone, \fairtwo)$-F$k$MO}] \label{def:fairkmedout}
Along with the finite set of points, $\clients$ and $\facilities$ and an integer $k$, we are given a partition of $\clients$ into $\groups$ groups, say $\LR{\clients_1, \ldots, \clients_{\groups}}$, two fairness vectors $\fairone,\fairtwo\in [0,1]^{\groups}$ and an outlier vector $\mathbbm{\outliers} \in \nat^{\groups}$. The objective is to find a subset $\solution \subseteq \facilities$ of size at most $k$, a set $\clients'$ of outliers and an assignment $\assign : \clients \setminus \clients' \rightarrow \solution$ such that:
\begin{enumerate}
    \item The number of clients of each color $g \in [\groups]$, in each cluster should be at most $\alpha_g$ and at least $\beta_g$ fraction of the total number of clients in the cluster. More formally, for each $g \in [\groups]$, and for each $f \in \solution$, the following holds: $\beta_g \le \frac{|\sigma^{-1}(f) \cap C_g|}{|\sigma^{-1}(f)|} \le \alpha_g.$
    \item For each color $g \in [\groups]$, the number of outliers is at most $m_g$, i.e., $\forall g \in [\groups], |c: c \in \clients' \cap \clients_g| \leq m_g$, and
     \item The following objective is minimized: $\cost(\assign)=\sum_{c \in \clients \setminus \clients'} \dist(c,\assign(c))$.
\end{enumerate}
\end{definition}
We give a brief sketch of the modifications required to adapt our approach for \ckmo in the context of \textsf{$(\fairone, \fairtwo)$-F$k$MO}. Let $\cI = ((\points, \dist), \clients, \facilities, k, \mathbbm{\outliers}, \alpha, \beta)$ denote the given instance of \textsf{$(\fairone, \fairtwo)$-F$k$MO}. Let $m \coloneqq \sum_{g \in [\groups]} m_g$ denote the total number of outliers across the colors. First, we define a suitable instance $I_0$ of $(k+m)$-\textsc{Median} by ignoring fairness constraints and by adding a co-located facility at each client. It can be easily that $\opt(I_0) \le \opt(\cI)$. Then, we obtain a constant-approximate solution for $I_0$, which is then used to perform ring sampling. The rings are defined as before, however, from each ring $\ringf{j}$ and for each color $g \in [\groups]$, we need to sample $O((m + k\log n)/\epsilon^3)$ points. Note that each weighted client in the sample retains its original color. It follows that the overall size of the sample $\wtset$ is bounded by $\lr{\frac{km\groups \log n}{\epsilon}}^{O(1)}$. In the enumeration phase, we try to guess the exact distribution of weighted outliers, and for each such guess we obtain an instance of \textsf{$(\fairone, \fairtwo)$-F$k$M} by deleting the corresponding set of guessed outliers. We use a $\gamma$-approximation algorithm on each such instance, and it can be argued that the best solution is a $(\gamma+\epsilon)$-approximation for $\cI$.

Now we sketch the analysis. First, given a feasible set $\solution\subseteq \facilities$, we need to compute a feasible assignment that respects the $(\alpha, \beta)$-fairness constraints. To this end, we need an analog of \mcfo, which we call \emph{Weighted Fair Assignment with Outliers} (\wfao) in the following subsection. As observed in \cite{BandyapadhyayFS21}, this problem is NP-hard, even without outliers; but they design an FPT algorithm for the problem. We generalize their ideas to handle outliers. \wfao is used in the algorithm to find a feasible assignment, as well as in the analysis to argue that the cost of various reassignments remains bounded. One key idea in our analysis, again, is to show that the cost of over/under-sampling from the set of outliers in each ring is negligible, since the size of the sample is large compared to $m$, the total number of outliers. In \Cref{subsec:fair-wfao} give a formal definition of \wfao and design an FPT algorithm for the problem. The rest of the analysis is similar to that for \ckmo, and is deferred to a later full version of the paper. We also note that as in the previous subsection, the arguments for \textsf{$(\fairone, \fairtwo)$-F$k$MO} readily generalize to \kmeans and \textsc{$k$-Facility Location} objectives. Furthermore, the whole approach also generalizes to the case when the groups of clients are not necessarily disjoint. We omit these details.

\subsection{Weighted Fair Assignment with Outliers} \label{subsec:fair-wfao}
For a fixed set of $k$ facilities $\solution=\{ f_1,f_2,\ldots,f_k\}$, we define $\wcost_{\outliers}(\wtset,M,\solution)$ as the cost of optimal assignment satisfying property 1-3 above over all $M \in \collection$. If there is no such assignment we call such $\solution$ and $M$ to be infeasible and define $\wcost_{\outliers}(\wtset,M,\solution) = \infty$. Note that, for given $\solution$ and $M$, $\wcost_{\outliers}(\wtset,M,\solution)$ can be obtained by solving a \mcfo. But, we can not try out all possible $M$ as the size of $\collection$ can be very large. To handle this, for feasible fixed $\solution$, we modify the weighted fair assignment algorithm presented in ~\cite{BandyapadhyayFS21} for finding the optimal assignment of $\wtset$ to $\solution$  satisfying the fairness and outlier constraints. Let us call this problem as {\em Weighted Fair Assignment with Outliers} problem (\wfao). Assignment problem for ($\fairone,\fairtwo$)-Fair Clustering and hence for ($\fairone,\fairtwo$)-Fair Clustering with Outliers is NP-Hard~\cite{Bera2019}. Therefore, we give a FPT algorithm with running time FPT in $k$ and $\groups$.

\begin{lemma}
	There exists an algorithm that returns integral solution for \wfao. The running time of algorithm is $(k\groups)^{O(k \groups)}n ^{O(1)} L$ where $L$ is the total number of bits in the encoding of distances and weights in the problem instance $\cI$.
\end{lemma}
\begin{proof}
	We formulate \wfao as a Mixed-Integer Linear Program. For every client $c\in \clients$ and every facility $f \in \solution$, let $x_{cf}$ be the variable denoting how much weight of client $c$ is assigned to facility $f$. For every facility $f \in \solution$ and every group $g \in \{1,\ldots, \groups \}$, let $y_{fg}$ be the variable denoting how much weight from clients of group $g$ are assigned to facility $f$. For every client $c \in \clients$, let $z_{c}$ denote how much weight of client $c$ is outlier. \wfao can now be formulated as Mixed-Integer Linear program as follows:
	
	$$\text{Minimize}~\cost(x,y,z)$$
	\begin{eqnarray}
		\text{subject to},~
		& \sum_{f \in \solution} x_{cf} + z_{c}  = \wt(c) & \forall c \in \clients \label{LP_const1}\\ 
		& \sum_{c: c \in \clients_g} z_{c} \leq \outliers_g & \forall g \in \{ 1,\ldots, \groups \}  \label{LP_const2}\\
		& \sum_{c: c \in \clients_g} x_{cf} = y_{fg} & \forall f \in \solution, \forall g \in \{ 1,\ldots, \groups \}\label{LP_const3}\\ 
		& \sum_{c:c \in \clients_g} x_{cf} \leq \fairone_g \sum_{c \in \clients} x_{cf} &  \forall f \in \solution, \forall g \in \{ 1,\ldots, \groups \}\label{LP_const4}\\
		& \sum_{c:c \in \clients_g} x_{cf} \geq \fairtwo_g \sum_{c \in \clients} x_{cf} &  \forall f \in \solution, \forall g \in \{ 1,\ldots, \groups \}\label{LP_const5}\\
		& x_{cf},z_{c} \geq 0 & \forall c \in \clients, \forall f \in \solution \label{LP_const6} \\
		& y_{fg} \in \nat & \forall c \in \clients , \forall g \in [\groups] \label{LP_const7}
	\end{eqnarray}
	where $\cost(x,y,z) = \sum_{c\in \clients} \sum_{f \in \facilities} \dist(c,f) x_{cf}$ in the case of $k$-Median and $\cost(x,y,z) = \sum_{c\in \clients} \sum_{f \in \facilities} \dist(c,f)^2 x_{cf}$ in the case of $k$-Means. Constraints~\ref{LP_const1} ensures that for any client $c$, the total of weight assigned to facilities in $\solution$ and the weight that is outlier is exactly equal to the weight of $c$. Constraints~\ref{LP_const2} ensures the bound on number of outliers for every group $g \in [\groups]$. Constraints~\ref{LP_const4} and~\ref{LP_const5} are fairness constraints.
	
	We solve the above Mixed-Integer Linear Program using the following proposition.
	\begin{proposition}[\cite{Lenstra,kannan1987,Frank1987}]
		There is an algorithm solving Mixed-Integer Linear Programming in time $O(a^{2.5a+o(p)}b^4L)$ and space polynomial in $L$, where $a$ is the number of integer variables, $b$ is the number of non-integer variables and $L$ is the bitsize of the given instance.
	\end{proposition}
	Note that, in our Mixed-Intger Linear Program, $a=k\groups$ and $b=n \outliers$. Therefore, in time $(k\groups)^{O(k \groups)}(n \outliers)^{O(1)} L$ we find the optimal solution $\{x_{cf}\}, \{y_{fg}\}, \{z_c\}$.
	
	$\{ y_{fg}\}$ are integral but $\{x_{cf}\}$ and $\{ z_c\}$ might be fractional. We next show that the integrality of $\{y_{fg}\}$ can be used to find another optimal solution that is integral in $\{x_{cf}\}$ and $\{ z_c\}$ also. For every group $g \in [\groups]$, consider the following network follow: for every client $c \in \clients_g$ there is a supply node with $\wt(c)$ supply and for every facility $f \in \solution$ we have a demand node with demand $y_{fg}$. We introduce edge between each client $c \in \clients_g$ and each facility $f \in \solution$ with unlimited capacity and $\dist(c,f)$ cost. Since all the demands and supplies are integral, there exist a maximum flow of minimum cost (say $\{x'_{cf}\}$) that is integral and this flow can be computed by using \mcf in polynomial time. Replacing $\{x_{cf}\}$ with $\{x'_{cf}\}$ gives us another optimal solution where $\{x_{cf}\}$ and $\{y_{fg}\}$ are integral. Observe that, once we have $\{x_{cf}\}$ to be integral, $\{z_c\}$ are also integral because $\wt(c)$ are integer and $z_c = \wt(c)-\sum_{f \in \solution}x_{cf}$ fro all $c \in \clients$. Once, we repeat the above procedure for every group $g \in [\groups]$, the solution is completely integral. 
\end{proof}

%% file: multiring.tex
\section{Analysis of Multiple Rings Case} \label{sec:multiring-full}
In \Cref{sec:singlering}, we have shown how to bound the error for a single ring case. We now consider the multiple rings where we use union bound over all the rings to obtain the desired claims. 

Consider the rings in any arbitrary order $\order$. For any two rings $\rings{f}{j}$ and $\rings{f'}{j'}$, we say $(f,j)<(f',j')$ if  $\rings{f'}{j'}$ comes after $\rings{f}{j}$ in $\order$. Fix a ring $\rings{f}{j}$. Now we define a function $\func_{f,j}$ for each ring similar to the function $\func$ defined in Section~\ref{sec:singlering}. As done in Section~\ref{sec:singlering}, to define $\func_{f,j}$, we first create an instance $\flow(\vectory)$ of \mcfo corresponding to a vector $\vectory$ of size $n$. To create this instance, we define a random vector $\vectorx \in \mathbb{R}_{+}^{|\rings{f}{j}|}$ where each coordinate pick value $\frac{|\rings{f}{j}|}{s}$ with probability $\frac{s}{|\rings{f}{j}|}$ and 0 otherwise. In $\flow(\vectory)$, every client $c \in \rings{f}{j}$ has $\vectorx_{c}$ demand and the cluster center $f$ has $|\rings{f}{j}|- \sum_{c \in \rings{f}{j}} \vectorx_{c}$ demand. But, as we consider the sample from the other rings too, $\func_{f,j}$ will also depend on these samples. Suppose we fix samples $\samples{f'}{j'}$ for every ring $\rings{f'}{j'}$, $(f',j') \neq (f,j)$. In $\flow(\vectory)$, set demand $\frac{|\rings{f'}{j'}|}{s}$ at each client $c \in \samples{f'}{j'}$. $\func_{f,j}(\vectory)$ is now the optimal cost of $\flow(\vectory)$.

Let $\ex_{\samples{f'}{j'} : (f',j') > (f,j)}[\func_{f,j}(\vectory) | \samples{f'}{j'}:(f',j') < (f,j)]$ or in short $\ex_{>(f,j)}[\func_{f,j}(\vectory)]$ be the expectation of $\func_{f,j}$ over all samples $\samples{f'}{j'}$ for $(f',j')>(f,j)$ given fixed samples $\samples{f'}{j'}$ for all $(f',j')<(f,j)$. Similarly, define $\ex_{\samples{f'}{j'} : (f',j') \geq (f,j)}[\func_{f,j}(\vectory) | \samples{f'}{j'}:(f',j') < (f,j)]$ or in short $\ex_{\geq(f,j)}[\func_{f,j}(\vectory)]$ be the expectation of $\func_{f,j}$ over all samples $\samples{f'}{j'}$ for $(f',j')\geq(f,j)$ given fixed samples $\samples{f'}{j'}$ for all $(f',j')<(f,j)$. 

Let $(f_1, j_1)$ and $(f_l, j_l)$ be the indexes of first and last rings in this order, respectively. Then, $\ex_{\ge(f_1, j_1)}[\func_{f, j}(\vectory)] = \cost_{\outliers}(\clients, \solution)$, and $\ex_{>(f_l, j_l)}[\func_{f, j}(\vectory)] = \wcost_{\outliers}(\wtset, \solution)$.

With these definitions, we get the following lemma analogous to \Cref{lem:part1} and \Cref{lem:part2} (combined) in single ring case. 

\begin{lemma}
    With probability at least $1-n^{-(k+\consta)}$, for any ring $\rings{f}{j}$, $|\ex_{>(f,j)}[\func_{f,j}(\vectory)] - \ex_{\geq(f,j)}[\func_{f,j}(\vectory)]| \leq \epsilon \constb |\rings{f}{j}| \radr$ where $\radr$ is the radius of ring $\rings{f}{j}$ and $\consta,\constb$ are constants.
\end{lemma}
Combining over all rings and using respective definitions of $\ex_{>(f,j)}[\wcost_\outliers(\wtset,F)]$ and $\ex_{\geq(f,j)}[\wcost_\outliers(\wtset,F)]$, we get the following lemma:

\begin{lemma}
\label{lem:combined}
    For any feasible $F \subseteq \facilities$, $|\ex_{>(f,j)}[\wcost_\outliers(\wtset,F)] - \ex_{\geq(f,j)}[\wcost_\outliers(\wtset,F)]| \leq \epsilon \constb |\rings{f}{j}| \radr$ with probability $1-n^{-(k+\consta)}$.
\end{lemma}

We take union bound over all possible set of feasible solutions which gives that inequality in \Cref{lem:combined} fails with probability $\leq n^{-c}$ and hence the lemma holds with high probability. Now consider the process of going through all the rings $\rings{f}{j}$ according to $\sigma$. 
Applying Lemma~\ref{lem:combined} on all $O((k+m) \log n) \le n^2$ rings conditioned on the choices of $\samples{f'}{j'}$ for $(f',j') < (f,j)$. We get the following with high probability,

\begin{equation}
\begin{split}
    |\wcost_\outliers(\wtset,F) - \ex[\wcost_\outliers(\wtset,F)]| & \leq \sum_{(f,j)} \epsilon \constb  |\rings{f}{j}| \radr \\ &  = 2 \epsilon \constb \cdot \sum_{(f,j)} |\rings{f}{j}| \cdot \radr / 2 \\ & \leq 2 \epsilon \constb \cdot \cost_0(\clients,\solution_{\factortwo}) \\ & \leq 2 \epsilon \constb \cdot \factortwo \cdot \opt(I) \\ & \leq O(\epsilon) \cost_m(C, F),
\end{split}   
\end{equation}

where the second last inequality follows from \Cref{eqn:alpha-opt-bd}. Note that, $\ex[\wcost_\outliers(\wtset,F)] = \cost_{\outliers}(\clients,F)$. Therefore, scaling down $\epsilon$ by a constant factor to get $\epsilon \cdot \cost_m(C, F)$ from $O(\epsilon) \cost_m(C, F)$ gives us Lemma~\ref{lem:coreset-multiring}.

%% file: conclusion.tex
\section{Conclusion} \label{sec:conclusion}

In this paper, we further demonstrated the power of ring sampling approach by showing that it is versatile enough to handle multiple constraints in a clustering problem. In particular, we designed FPT approximations for \ckmo. We further generalized the arguments to handle outliers while maintaining ($\fairone,\fairtwo$)-fairness. There are several questions left open from our work: which other simultaneous orthogonal constraints can we handle at the same time? Can we extend this approach to other kinds of clustering objectives, such as sum-of-radii clustering? Perhaps, more closely related to result in this paper, another open question is whether one can improve the approximation guarantee for \ckmo (even for a special case when the capacities are uniform)? One way to do is to improve the approximation for \ckm.